\documentclass[a4paper,UKenglish,cleveref, autoref]{lipics-v2021}

\usepackage[utf8]{inputenc}
\usepackage[appendix=inline]{apxproof} %
\usepackage{hyperref}
\usepackage[color=yellow!20]{todonotes}
\usepackage{graphicx}
\usepackage{xspace}
\usepackage{marginnote}
\usepackage{tikz}
\usepackage{subcaption}

\renewcommand{\mainbodyrepeatedtheorem}{\textup{($\star$) \marginnote{($\star$): Proof in appendix}}}

\newtheoremrep{theorem}{Theorem}
\newtheoremrep{lemma}[theorem]{Lemma}
\newtheoremrep{observation}[theorem]{Observation}

\renewcommand\leq\leqslant
\renewcommand\geq\geqslant

\newcommand{\LRS}{{\sc Longest Run Subsequence}\xspace}

\newcommand{\TwoLRS}{{\sc 2-Longest Run Subsequence}\xspace}

\newcommand{\MISC}{{\sc MISC}\xspace}

\newcommand{\Pb}[4]{%
\begin{center}
  \begin{tabular}{|l|}%
  \hline
    \begin{minipage}[c]{\textwidth}
      \smallskip%
      \par\noindent%
      #1%
      \par\noindent%
      $\bullet$
      \textbf{\textsf{Input}}: #2%
      \par\noindent%
      $\bullet$
      \textbf{\textsf{#4}}: Does there exist %
      #3?%
      \smallskip%
      \par\noindent%
    \end{minipage}
  \\\hline
  \end{tabular}%
\end{center}
}%

\title{The Longest Run Subsequence Problem: Further Complexity Results}

\author{Riccardo Dondi}{Universit\`a degli Studi di Bergamo, Bergamo, Italy}{riccardo.dondi@unibg.it}{https://orcid.org/0000-0002-6124-2965}{}%

\author{Florian Sikora}{Universit\'{e} Paris-Dauphine, PSL University, CNRS, LAMSADE, 75016 Paris, France}{florian.sikora@dauphine.fr}{https://orcid.org/0000-0003-2670-6258}{Partially funded by ESIGMA (ANR-17-CE23-0010)}

\authorrunning{R. Dondi and F. Sikora} %

\Copyright{Riccardo Dondi and Florian Sikora} %



\category{} %

\nolinenumbers %
\hideLIPIcs  %

\EventEditors{Pawe{\l} Gawrychowski and Tatiana Starikovskaya}
\EventNoEds{2}
\EventLongTitle{32nd Annual Symposium on Combinatorial Pattern Matching (CPM 2021)}
\EventShortTitle{CPM 2021}
\EventAcronym{CPM}
\EventYear{2021}
\EventDate{July 5--7, 2021}
\EventLocation{Wroc{\l}aw, Poland}
\EventLogo{}
\SeriesVolume{191}
\ArticleNo{9}

\date{}

\begin{document}

\maketitle

\begin{abstract}
\LRS is a problem introduced recently  in the
context of the scaffolding phase of genome assembly (Schrinner et al., WABI 2020). 
The problem asks for a maximum length subsequence of a given
string that contains at most one run for each symbol (a run is a maximum substring of consecutive identical symbols).
The problem has been shown to be NP-hard and 
to be fixed-parameter tractable when the parameter
is the size of the alphabet on which the input string
is defined. 
In this paper we further investigate the complexity of
the problem and we show that it is fixed-parameter
tractable when it is parameterized by the number of runs
in a solution, a smaller parameter. 
Moreover, we investigate the kernelization
complexity of \LRS and we prove that it does not
admit a polynomial kernel when parameterized
by the size of the alphabet or by the number of runs.
Finally, we consider the restriction of \LRS when 
each symbol has at most two occurrences in the input
string and we show that it is APX-hard.
\end{abstract}

\section{Introduction}

A fundamental problem in computational genomics
is genome assembly, whose goal is
reconstructing a genome given a set of reads (a read is a sequence of base pairs)~\cite{Alonge519637,bioinfo2020}. 
After the generation of initial assemblies, called \emph{contigs}, they have
to be ordered correctly, in a phase called \emph{scaffolding}. 
One of the commonly used approaches for scaffolding
is to consider two (or more) incomplete
assemblies of related samples, thus allowing
the alignment of contigs based on their similarities \cite{Goel546622}.
However, the presence of genomic repeats and
structural differences may lead %
to misleading connections between contigs. 

Consider two sets $X$, $Y$ of contigs, 
such that the order of contigs in $Y$ has
to be inferred using the contigs in $X$. 
Each contig in $X$ is divided into equal size 
bins and each bin is mapped to a contig in $Y$ (based
on best matches). 
As a consequence, each bin in $X$ can be partitioned based
on the mapping to contigs of $Y$.
However this mapping of bins to contigs,
due to errors (in the sequencing or in the mapping
process) or mutations, 
may present some inconsistencies,
in particular bins can be mapped to scattered contigs,
thus leading to an inconsistent partition of $X$,
as shown in Fig. \ref{fig:Example}.
In order to infer the most likely partition of $X$
(and then distinguish between the transition
from one contig to the other and errors in the
mapping),  the method proposed in \cite{DBLP:conf/wabi/SchrinnerGWSSK20} 
asks for a longest subsequence of the contig matches in $X$ such that each contig run occurs at most 
once (see Fig. \ref{fig:Example} for an example). 

\begin{figure}
\centering
\includegraphics[scale=.4]{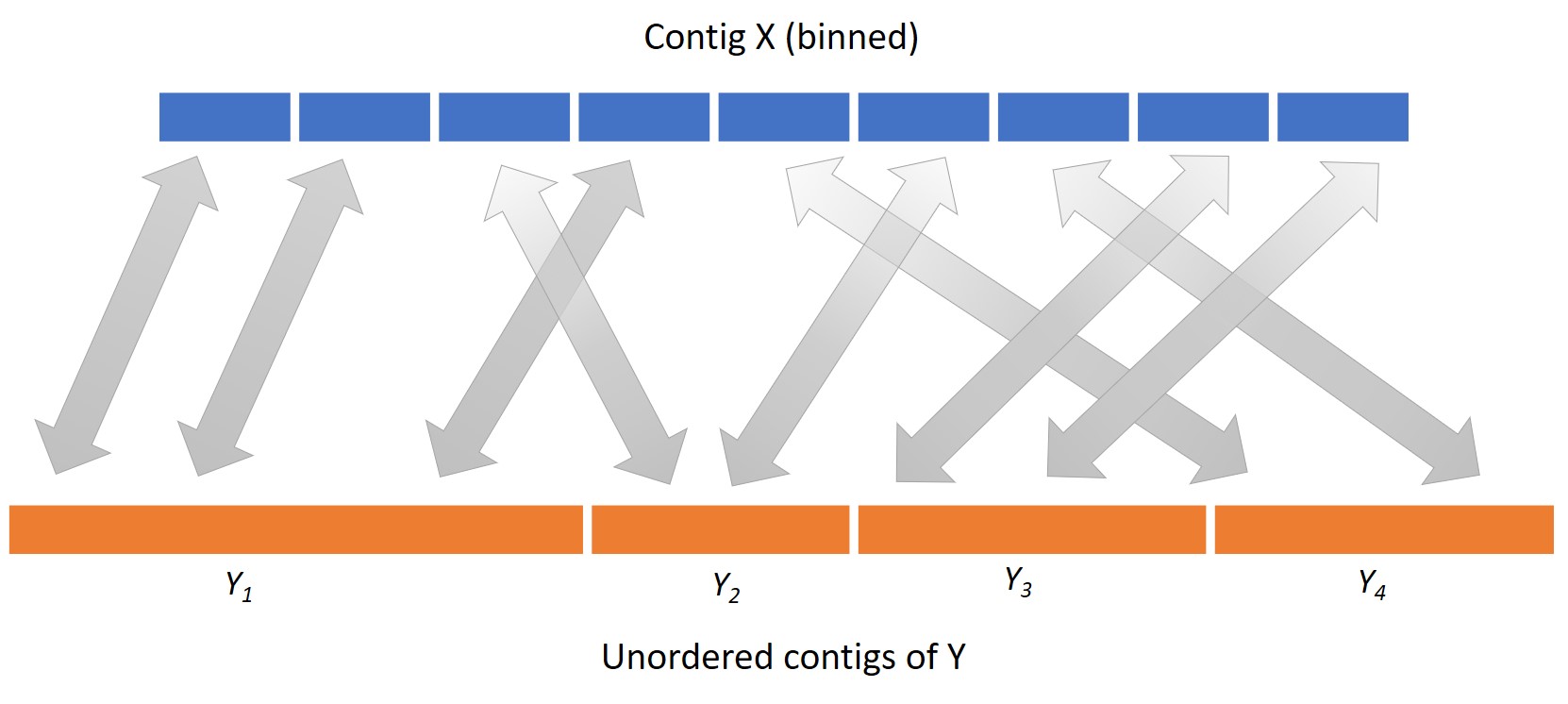}
\caption{
An example of matching
a binned contig ($X$)
with the unordered
contigs of $Y$.
The string inferred
from this matching is 
$S =y_1~y_1~y_2~y_1~y_4~y_2~y_4~y_3~y_3$.
Notice that $S$ induces an inconsistent
partition of the bins of $X$, for example for the mapping
of $Y_1$ and $Y_2$. Indeed, $Y_1$ is mapped in the first, second and fourth bin of $X$, while $Y_2$ is mapped 
in the third and sixth bin of $X$.
A longest run subsequence $R$ of $S$
is $R = y_1~y_1~y_1~y_4~y_4~y_3~y_3$, that induces
a partition of some bins of $X$.
}
\label{fig:Example}
\end{figure}

This problem, called \LRS, has been recently introduced 
and studied by Schrinner et al.~\cite{DBLP:conf/wabi/SchrinnerGWSSK20}.
\LRS has been shown to be 
NP-hard~\cite{DBLP:conf/wabi/SchrinnerGWSSK20}
and fixed-parameter tractable when the parameter
is the size of the alphabet on which the input string
is defined~\cite{DBLP:conf/wabi/SchrinnerGWSSK20}.
Furthermore, an integer linear program has
been given for the problem~\cite{DBLP:conf/wabi/SchrinnerGWSSK20}.
Schrinner et al. let as future work approximability and parameterized complexity results on the problem~\cite{DBLP:conf/wabi/SchrinnerGWSSK20}.
Note that this problem could be seen as close to the "run-length encoded" string problems in the string literature, where a string is described as a sequence of symbols followed by the number of its consecutive occurrences, i.e. a sequence of runs where only its symbol and its length is stored (see for example~\cite{DBLP:journals/jc/ApostolicoLS99}).
While finding the longest common subsequence between two of such strings is a polynomial task, our problem is, to the best of our knowledge, not studied in literature before the work of Schrinner et al.~\cite{DBLP:conf/wabi/SchrinnerGWSSK20}.

In this paper we further investigate the complexity of
the \LRS problem.
We start in Section~\ref{sec:definition} by introducing
some definitions and by giving the formal definition
of the problem.
Then in Section \ref{sec:FPT} we give a randomized
fixed-parameter algorithm, where the parameter is the number of runs
in a solution, based on the multilinear detection technique. In Section \ref{sec:kernelization}, we investigate the kernelization
complexity of \LRS and we prove that it does not
admit a polynomial kernel when parameterized
by the size of the alphabet or by the number of runs.
Notice that the problem admits a polynomial kernel
when parameterized by the length of the solution
(see Observation \ref{obs:sp}).
Finally, in Section \ref{sec:APX2occ} we consider the restriction of \LRS when 
each symbol has at most two occurrences in the input
string and we show that it is APX-hard.
We conclude the paper with some open problems.

\section{Definitions}
\label{sec:definition}

In this section we introduce the main definitions
we need in the rest of the paper.

\subparagraph*{Problem Definition}
Given a string $S$, 
$|S|$ denotes the length of the string; 
$S[i]$, with $1 \leq i \leq |S|$,
denotes the symbol of $S$ in position $i$,
$S[i,j]$,
with $1 \leq i \leq j \leq |S|$, denotes
the substring of $S$ that starts in position $i$ 
and ends in position $j$. Notice that if $i=j$,
then $S[i,i]$ is the symbol $S[i]$.
Given a symbol $a$, we denote by $a^p$, for some integer
$p \geq 1$, a string consisting of the concatenation of 
$p$ occurrences
of symbol $a$.

A \emph{run} in $S$ is a substring
$S[i,j]$, with $1 \leq i \leq j \leq |S|$,
such that $S[z] = a$, for each $i \leq z \leq j$,
with %
$a \in  \Sigma $.
Given $a \in \Sigma$, an $a$-run is a run in $S$
consisting of repetitions of symbol $a$.
Given a string $S$ on alphabet $\Sigma$, a \emph{run subsequence} $S'$
of $S$ is a subsequence that contains \emph{at most
one run} for each symbol $a \in \Sigma$.

Now, we are ready to define the \LRS problem.

\Pb{\LRS}{A string $S$ on alphabet $\Sigma$, an integer $k$.}{a run subsequence $R$ of length $k$}{Output}

A string $S[i,j]$ contains an $a$-run
with $a \in \Sigma$, 
if it
contains a run subsequence which is an $a$-run. 
A run subsequence of $S$ which is an $a$-run, 
with $a \in \Sigma$, is maximal  if it contains
all the occurrences of symbol $a$ in $S$.
Note that an optimal solution
may not take maximal runs.
For example, 
consider 
\[S = a b a c  a a  b b a b\]
an optimal run subsequence in $S$ is
\[R = a a a a b b b\] 
Note that no run in $R$ is maximal and even that some symbol of $\Sigma$ may not be in an optimal solution of \LRS{}, in the example no 
$c$-run belongs to $R$.

\subparagraph*{Graph Definitions}

Given a graph $G=(V,E)$, we 
denote by $N(v) = \{ u: \{ u,v\} \in E  \}$,
the neighbourhood of $v$.
The closed neighbourhood of $v$ is $N[v] = N(v) \cup \{v\}$.
$V' \subset V$ is an independent set
when $\{ u,v \} \notin E$ for each $u,v \in V'$.
We recall that a graph $G=(V,E)$ is \emph{cubic} when 
$|N(v)|=3$ for each $v \in V$.

\subparagraph*{Parameterized Complexity}

A \emph{parameterized problem} is a decision problem specified together with a \emph{parameter}, that is, an integer $k$ depending on the instance.
A problem is \emph{fixed-parameter tractable} (FPT for short) if it can be solved in time $f(k)\cdot|I|^c$ (often briefly referred to as FPT\emph{-time}) for an instance $I$ of size $|I|$ with parameter $k$, where $f$ is a computable function and $c$ is a constant. %
Given a parameterized problem $P$, a \emph{kernel} is a polynomial-time computable function which associates with each instance of $P$ an equivalent instance of $P$ whose size is bounded by a function $h$ of the parameter. When $h$ is a polynomial, the kernel is said to be \emph{polynomial}. 
See the book~\cite{CFK+15} for more details.

In order to prove that such polynomial kernel is unlikely, we need additional definitions and results. %

\begin{definition}[Cross-Composition~\cite{Bodlaender2014}]
\label{def:CrossComp}
We say that a problem $L$ \emph{cross-composes} to a parameterized problem $Q$ if there is a polynomial equivalence relation $\mathcal{R}$ and an algorithm which given $t$ instances $x_1,x_2,\ldots,x_t$ of $L$ belonging to the same equivalence class $\mathcal{R}$, computes an instance $(x^*,k^*)$ of $Q$ in time polynomial in $\sum_{i=1}^t|x_i|$ such that (i) $(x^*,k^*)\in Q \iff x_i \in L$ for some $i$ and (ii) $k^*$ is bounded by a polynomial in $(\max_i|x_i|+\log t)$.
\end{definition}

This definition is useful for the following result, which we will use to prove that a polynomial kernel for \LRS with parameter $|\Sigma|$ is unlikely.

\begin{theorem}[\cite{Bodlaender2014}]
\label{teo:CrossComp}
If an NP-hard problem $L$ has a cross-composition into a parameterized problem $Q$ and $Q$ has a polynomial kernel, then NP $\subseteq$ coNP/poly.
\end{theorem}

For our FPT algorithm, we will reduce our problem to another problem, called $k$-\textsc{Multilinear Detection} problem ($k$-MLD), which can be solved efficiently.
In this problem, we are given a polynomial over a set of variables $X$, represented as an arithmetic circuit $\mathcal{C}$, and the question is to decide if this polynomial contains a multilinear term of degree exactly $k$.
A polynomial is a sum of monomials.
The \emph{degree} of a monomial is the sum of its variables degrees and a monomial is \emph{multilinear} if the degree of all its variables is equal to 1 (therefore, a multilinear monomial of degree $k$ contains $k$ different variables).
For example, $x_1^2x_2+x_1x_2x_3$ is a polynomial over 3 variables, both monomials are of degree 3 but only the second one is multilinear.

Note that the size of the polynomial could be exponentially large in $|X|$ and thus we cannot just check each monomial.
We will therefore encode the polynomial in a compressed form:
the circuit $\mathcal{C}$ is represented as a Directed Acyclic Graph (DAG), where leaves are variables $X$ and internal nodes are multiplications or additions. 
The following result is fundamental for $k$-MLD.

\begin{theorem}[\cite{DBLP:conf/icalp/Koutis08,DBLP:journals/ipl/Williams09}]\label{thm:kmld}
There exists a randomized algorithm solving $k$-MLD in time $O(2^k|C|)$ and $O(|C|)$ space.
\end{theorem}

\subparagraph*{Approximation}

In Section \ref{sec:APX2occ}, we prove the APX-hardness
of \LRS with at most two occurrences
for each symbol in $\Sigma$,
by designing an L-reduction
from {\sc Maximum Independent Set} on cubic graphs.
We recall here the definition of L-reduction.
Notice that, given a solution $S$ of a problem 
($A$ or $B$ in the definition), we denote
by $val(S)$ the value of $S$
(for example, in our problem, the length of
a run subsequence).

\begin{definition}[L-reduction \cite{DBLP:journals/jcss/PapadimitriouY91}]
\label{def:L-reduction}
Let $A$ and $B$ be two optimization problems. 
Then $A$ is said to be L-reducible to $B$ if there are two constants $\alpha, \beta > 0$ and two polynomial-time computable functions $f,g$ such that: 
(i) $f$ maps an instance $I$ of $A$ into an instance $I'$ of $B$ such that $opt_B(I') \leq \alpha \cdot opt_A(I)$,
(ii) $g$ maps each solution $S'$ of $I'$ into a solution $S$ of $I$ such that $|val(S)-opt_A(I)| \leq \beta \cdot |val(S')-opt_B(I')|$.
\end{definition}

L-reductions are useful in order to apply the following theorem.

\begin{theorem}[\cite{DBLP:journals/jcss/PapadimitriouY91}]\label{thm:noptas}
Let $A$ and $B$ be two optimization problems. 
If $A$ is L-reducible to $B$ and $B$ has a PTAS, then $A$ has a PTAS.
\end{theorem}

\subsection*{Parameterized Complexity Status of the Problem}

In the paper, we consider the parameterized complexity
of \LRS{} under the different parameterizations.
We consider the following parameters:

\begin{itemize}
\item The length $k$ of the solution of \LRS %

\item The size $|\Sigma|$ of the alphabet 

\item The number $r$ of runs in a solution of \LRS

\end{itemize}

Notice that $r \leq |\Sigma| \leq k$.
Indeed, there always exists a solution consisting
of one occurrence for each symbol in $\Sigma$,
hence we can assume that $|\Sigma| \leq k$.
Clearly, $r \leq |\Sigma|$, since
each run in a solution of \LRS is associated with a 
distinct symbol of $\Sigma$.

In Table~\ref{tab:my_label}, we present the status of the parameterized complexity of \LRS for these parameters.

\begin{table}[h]
    \centering
    \begin{tabular}{c|c|c}
                    & FPT   & Poly Kernel \\\hline
        $k$      &  Yes (Obs. \ref{obs:sp})    & \textbf{Yes} (Obs. \ref{obs:sp}) \\
        $|\Sigma|$  & Yes~\cite{DBLP:conf/wabi/SchrinnerGWSSK20}      &  \textbf{No} (Th. \ref{th:nokernel}) \\
        $r$         & \textbf{Yes \& Poly Space}~(Th.~\ref{teo:LRSTeo} )       & \textbf{No} (Cor. \ref{cor:nokernel}) \\

    \end{tabular}
    \caption{Parameterized complexity status for the three different parameters considered in this paper. Since these parameters are in decreasing value order, note that positive results propagate upwards, while negative results propagate downwards.
    In bold the new results we present in this paper.}
    \label{tab:my_label}
\end{table}

It is easy to see that \LRS{} has a polynomial kernel
for parameter $k$.

\begin{observationrep}
\LRS has a $k^2$ kernel.
\label{obs:sp}
\end{observationrep}
\begin{proof}
First, notice that if there exists an $a$-run $R'$ of length 
at 
least $k$, for some $a \in \Sigma$, then $R'$ is a
solution of \LRS.
Also note that if $|\Sigma| \geq k$, let $R^+$
be a subsequence of $S$ consisting of one occurrence
of each symbol of $\Sigma$ (notice that it is always
possible to define such a solution).
Then $R^+$ is a solution of \LRS of sufficient size.

Therefore, we can assume that $S$ is defined
over an alphabet $|\Sigma| < k$ and that 
each symbol has less then $k$ occurrences (otherwise
there exists an $a$-run of length at least $k$ for some $a \in \Sigma$). Hence \LRS has a kernel of size $k^2$.
\end{proof}

Schrinner et al. prove that \LRS is in FPT for parameter $|\Sigma|$, using exponential space~\cite{DBLP:conf/wabi/SchrinnerGWSSK20}. 
Due to a folklore result~\cite{CFK+15}, %
this also implies that there is a kernel for this parameter.
We will prove that there is no polynomial kernel for this parameter in Section~\ref{sec:kernelization}.

\section{An FPT Algorithm for Parameter Number of Runs}
\label{sec:FPT}

In this section, we consider \LRS{} when
parameterized by the number of different runs,
denoted by $r$, in the solution, that is 
whether there exists
a solution of \LRS{} consisting of exactly $r$ runs
such that it has length at least $k$.
We present a randomized fixed-parameter
algorithm for \LRS 
based on multilinear monomial detection.

The algorithm we present is for a variant of 
\LRS that asks for a run subsequence $R$ of $S$ such that
(1) $|R| = k$ and (2) $R$ contains exactly $r$ runs.
In order to solve the general problem where
we only ask for a solution of length at least $k$,
we need to apply the algorithm
for each $k$, with $r \leq k \leq |S|$. 

Now, we describe the circuit on which
our algorithm is based on.
The set of variables is:
\[
\{ x_a : a \in \Sigma \}
\]
Essentially, $x_a$ represents the fact
that we take an $a$-run (not necessarily maximal) in a substring of $S$. %

Define a circuit $\mathcal{C}$ as follows.
It has a root $P$ and a set
of intermediate vertices $P_{i,l,h}$,
with $1 \leq i \leq |S|$,
$1 \leq l \leq r$ and $1 \leq h \leq k$.
The multilinear monomials of $P_{i,l,h}$ informally encode a run subsequence
of $S[1,i]$ having length $h$ and consisting 
of $l$ runs.
$P_{i,l,h}$ is recursively defined as follows:

\begin{equation}
P_{i,l,h} = \left\{
\begin{array}{ll}
P_{i-1,l,h} + 
\sum_{j: 1 \leq j \leq i-1} P_{j,l-1,h-z} x_a    & \text{if } i \geq 1, l \geq 1,  h\geq 1, 
\\ 
& 1 \leq z \leq i-j-1,  \text{$S[i]=a$, $a \in \Sigma$,}\\
& 
\text{and } S[j+1,i] \text{ contains an $a$-run}\\
& 
\text{of length $z$}
\\
1 & \text{if } i \geq 0, l= h = 0, \\
0 & \text{if } i = 0, l > 0 \text{ or } h > 0.
\end{array} \right.
\label{Eq:recurrence}
\end{equation}

Then, define $P = P_{|S|,r , k} $.

Next, we show that we can consider the circuit $\mathcal{C}$ to compute a run subsequence of $S$. %

\begin{lemma}
\label{lemma:colorcoding}
There exists a run subsequence of $S$ of length $h$ consisting of $l$ runs over symbols $a_1, \dots a_l$ if and only if there exists a multilinear monomial in  $\mathcal{C}$ consisting of $l$ monomials $x_{a_1}, \dots , x_{a_l}$. 
\end{lemma}
\begin{proof}
We will prove that there is a run subsequence of $S$ of length $k$ consisting of $l$ runs over symbols $a_1, \dots a_l$ if and only if there exists a multilinear monomial in $\mathcal{C}$ of degree $l$, consisting of $l$ distinct variables $x_{a_1}, \dots , x_{a_l}$. 
In order to prove this result, we prove by induction on $i$ , $1 \leq i \leq |S|$, that there exists a run subsequence $R$ of $S[1 . . . i]$, such that $|R| = h$ and $R$ contains $l$ runs, an $a_z$-run for each
$a_z \in \Sigma$, $1 \leq z \leq l$,
if and only if there exists a multilinear monomial 
$x_{a_1} \dots   x_{a_l}$ in $P_{i, l,h }$.

\smallskip

We start with the case $i = 1$. 
Assume that there is a run subsequence consisting of a single run of length $1$ (say an $a_1$-run). 
It follows that $S[1] = a_1$ and, by Equation \ref{Eq:recurrence}, 
$P_{1,1,1} = P_{0,0,0} \cdot x_{a_1} =  x_{a_1} $.
Conversely, if $P_{1,1,1} = P_{0,0,0} \cdot x_{a_1} = x_{a_1} $, then by construction $S[1]  = a_1$, which is a run of length $1$.

\smallskip

Assume that the lemma holds for $j < i$, we prove that it holds for $i $. 

$(\Rightarrow)$ Assume that there exists a run subsequence $R$ of $S[1,i]$ that consists of $l$ runs and that has length $h$. 
Let the $l$ runs in $R$ be over symbols $a_1, \ldots , a_l$ and assume that the rightmost run in $R$ is an $a_l$-run.
If $S[i]$ does not belong to the $a_l$-run in $R$, then $R$ is a run subsequence in $S[1,i-1]$ and by induction hypothesis $P_{i-1,l,h}$ contains a multilinear monomial of 
degree $l$ over variables $x_{a_1} \dots   x_{a_l}$.
If $S[i]$ belongs to the $a_l$-run in $R$, then consider the $a_l$-run in $R$ and assume that it belongs to 
substring $S[j+1,i]$ of $S$, 
with $1  \leq j+1 \leq i$, and that it has length $z$.
Consider the run subsequence $R'$ of $S$ obtained from $R$ by removing the $a_l$-run.
Then, $R'$ is a run subsequence of $S[1,j]$ that does not contain $a_l$ (hence it contains $l-1$ runs) and has length $h-z$.
By induction hypothesis, $P_{j,l-1,h-z}$ contains a multilinear monomial of length $l-1$ over variables $x_{a_1} \dots   x_{a_{l-1}}$.
Hence by the first case of Equation \ref{Eq:recurrence}, it follows that $P_{i,l,h}$ contains a multilinear monomial of length $l$ over variables $x_{a_1} \dots   x_{a_l}$.

$(\Leftarrow)$ Assume that $P_{i,l,h}$ contains a multilinear monomial of length $l$ over variables $x_{a_1} \dots   x_{a_l}$, we will prove that there is a run 
subsequence of $S$ of length $k$ consisting of $l$ runs. 
By Equation \ref{Eq:recurrence}, it follows that (1) $P_{i-1,l,h}$ contains a multilinear monomial of length $l$ over variables $x_{a_1} \dots   x_{a_l}$ or (2) $P_{j,l-1,h-z}$, for some $1 \leq j \leq i-1$, contains a multilinear monomial  of length $l-1$ that does not contain one of  $x_{a_1} \dots   x_{a_l}$ (without loss
of generality $x_{a_l}$) and  $S[j+1,i]$ contains an $a_l$-run of length $z$.

In case (1), by induction hypothesis
there exists a run subsequence in $S[1,i-1]$
(hence also in $S[1,i]$) of length $h$ consisting of $l$ runs over symbols $a_1, \dots , a_l$.

In case (2), by induction hypothesis there exists a run 
subsequence $R'$ of $S[1,j]$ of length $h-z$ consisting
of $l-1$ runs over symbols $a_1, \dots, a_{l-1}$.
Now, by concatenating $R'$ with the $a_l$-run of
length $z$ in $S[j+1,i]$, we obtain a
run subsequence of $R$ of $S[1,i]$ consisting
of $l$ runs and having length $h$.
\end{proof}

\begin{theorem}
\label{teo:LRSTeo}
\LRS can be solved by a randomized algorithm in 
$O(2^r r |S|^3)$ time  and polynomial space. 
\end{theorem}
\begin{proof}
The correctness of the randomized algorithm
follows by Lemma~\ref{lemma:colorcoding}.

We compute $P$ in polynomial time and we decide if $P_{|S|,r,k}$ contains a multilinear monomial of degree $r$ in $O(2^r r |S|^2)$  time 
and polynomial space. 
The result follows from Lemma \ref{lemma:colorcoding}, Theorem \ref{thm:kmld},
and from the observation that
$|\mathcal{C}| = |S|\cdot l\cdot r$, with $l \leq |S|$.
Finally, we have to iterate the algorithm
for each $k$, with $r \leq k \leq |S|$,
thus the overall time complexity is $O(2^r r |S|^3)$.
\end{proof}

\section{Hardness of Kernelization}
\label{sec:kernelization}

As discussed in Section~\ref{sec:definition}, \LRS has a trivial polynomial kernel for parameter $k$ and its FPT status implies an (exponential) kernel for parameters $|\Sigma|$ and $r$. 
In the following, we will prove that it is unlikely that \LRS admits a polynomial kernel for parameter $|\Sigma|$ and parameter $r$. 

\begin{theorem}
\label{th:nokernel}
\LRS does not admit a polynomial kernel for parameter $|\Sigma|$, unless NP $\subseteq$ coNP/poly.
\end{theorem}
\begin{proof}

\sloppy
We will define an OR-cross-composition 
(see Definition \ref{def:CrossComp})
from the 
\LRS problem itself, whose unparameterized version is NP-Complete~\cite{DBLP:conf/wabi/SchrinnerGWSSK20}.

Consider $t$ instances $(S_1,\Sigma_1,k_1), (S_2,\Sigma_2,k_2), \ldots,  (S_t,\Sigma_t,k_t)$ of \LRS, where, for each $i$ with $1 \leq i \leq t$, $S_i$ is the input string built over the alphabet $\Sigma_i$, and $k_i \in \mathbb{N}$ is the length 
of the solution, respectively. 
We will define an equivalence relation $\mathcal{R}$ such
that strings that are not encoding valid instances 
are equivalent, and two valid instances 
$(S_i,\Sigma_i,k_i),(S_j,\Sigma_j,k_j)$ are equivalent if and only if $|S_i|=|S_j|$, $|\Sigma_i|=|\Sigma_j|$, and $k_i=k_j$. 
We now assume that $|S_i| = n$, $|\Sigma_i|=m$ 
and $k_i=k$ for all $1 \leq i \leq t$. 

We will build an instance of \LRS $(S',k', \Sigma')$ where $S'$ is a string built over the alphabet $\Sigma'$ and $k'$ an integer such that there is a solution of size at least $k'$ for $S'$ iff there is an $i$, $1 \leq i \leq t$ such that there is a solution of size at least $k$ in $S_i$. %

We first show how to redefine the input strings 
$S_1$, $S_2$, $\dots$, $S_t$,
such that they are all over the same alphabet. Notice that this will not be an issue, since we will construct a string $S'$ such that a solution of \LRS is not spanning over two different input strings.
For all instances $(S_i,\Sigma_i,k_i)$, 
$1 \leq i \leq t$,
we consider any ordering of the symbols in $\Sigma_i$ and we define a string $\sigma(S_i)$ starting from $S_i$,
by replacing the $j$-th, $1 \leq j \leq m$, symbol of $\Sigma_i$ by $j$, that is its position in the ordering of $\Sigma_i$. 
That way, it is clear that all strings $\sigma(S_i)$, 
$1 \leq i \leq t$ , are built over the same alphabet $\{1,2,\ldots,m\}$.

Now, the instance $(S',k', \Sigma')$ of \LRS is build as follows.
First, $\Sigma'$ is defined as follows:
\[
\Sigma' = \{1,2,\ldots,m\} \cup \{\#,\$\}
\]
where $\#$ and $\$ $ are two symbols not in $\Sigma$.

The string $S'$ is defined as follows:
\[
S' = \$^{2n} \sigma(S_1) \#^{2n} \$^{2n} \sigma(S_2) \#^{2n}, ..., \$^{2n}\sigma(S_t) \#^{2n}
\]
where $\$^{2n}$ ($\#^{2n}$, respectively) is 
a string consisting of the repetition $2n$ times of the symbol $\$$ ($\#$, respectively).

Finally, 
\[
k' = k + (t+1)(2n).
\]

Since we are applying OR-cross-composition for parameter
$|\Sigma|$, we need to show that property (ii) of
Definition \ref{def:CrossComp} holds.
By construction, we see that $|\Sigma'| = m+2$, which is independent of $n$ and $t$
and it is bounded by the size
of the largest input instance.

We now show that $S'$ contains a run subsequence of size at least $k'$ if and only if there exists at least one 
string $S_i$, $1 \leq i \leq t$, that
contains a run subsequence of size at least $k$.

$(\Leftarrow)$ First, assume that some $S_i$, with $1 \leq i \leq t$, 
contains a run subsequence $R_i$ of length at least $k$.
Then, define the following run subsequence $R'$ of $S'$, obtained by concatenating these
substrings of $S'$:
\begin{itemize}

\item The concatenation of the leftmost $i$-th 
substrings $\$^{2n}$ of $S'$,

\item The substring $\sigma(R_i)$ of $S_i$,

\item 
The concatenation of the rightmost $(t-i+1)$-th 
substrings $\#^{2n}$ of $S'$.

\end{itemize}

It follows that $S'$ contains a run subsequence of length at least $i \times (2n) + k + (t-i+1) \times (2n) = k'$.

$(\Rightarrow)$ Conversely, assume now that $S'$ contains a run subsequence $R'$ of length 
at least $k'$.
First, we prove %
that $R'$ contains exactly one $\$$-run
and one $\#$-run. Indeed, if it is not the case, we
can add the leftmost (the rightmost, respectively)
substring $\$^{2n}$ ($\#^{2n}$, respectively) as a run of $R'$.

Consider a run $r$ in $R'$,
which is either the $\$$-run or the $\#$-run of $R'$.
Assume that $R'$ contains a substring 
\[
r R'(S_i) R'(S_j)
\]
or a substring
\[
R'(S_i) R'(S_j)\ r
\]
with $1 \leq i < j \leq t$,
where $R'(S_i)$ ($R'(S_j)$, respectively) 
is a substring of $\sigma(S_i)$ 
(of $\sigma(S_j)$, respectively).
We consider without loss of generality 
the case that $r R'(S_i) R'(S_j)$ is a
substring of $R'$.
Then, we can modify $R'$, increasing
its length, as follows: 
we remove $R'(S_i)$ and extend
the run $r$ with a string $\$^{2n}$ or a string $\#^{2n}$ (depending on the fact
that $r$ is a $\$$-run or a $\#$-run, respectively)
that is between $\sigma(S_i)$ and $\sigma(S_j)$. The size
of $R'$ is increased, since $|R'(S_i)| \leq n$.

Now, assume that $R'$ contains a substring 
\[
r = \#^{2n}\ R'(S_i)\ \$^{2n} 
\]
where $R'(S_i)$ is a substring of $\sigma(S_i)$.
We can replace $r$ with the substring $\#^{2n} \#^{2n} \$^{2n} $, where $\#^{2n}$ is the substring
between $\sigma(S_i)$ and $\sigma(S_{i+1})$
in $S'$.
Again, the size of $R'$ is increased, since $|R'(S_i)| \leq n$.

By iterating these modifications on $R'$, we obtain
that $R'$ is one of the following string: 

\begin{enumerate}

\item A prefix $\$^{j(2n)}$ concatenated with 
a substring $R'(S_j)$ of $\sigma(S_j)$, 
for some $1 \leq  j \leq t$, concatenated with 
a suffix  $\#^{2n(t-j+1)}$

\item  A substring $R'(S_1)$ of $\sigma(S_1)$
concatenated with a string of $\#^{2nj}$ concatenated with %
a string of  $\$^{2n(t-j)}$ concatenated 
with a substring  $R'(S_t)$ of $\sigma(S_t)$.

\end{enumerate}

Notice that in this second case, it holds that
$|R'| = (t \times 2n) + |R'(S_1)| + |R'(S_t)| < k'$, since 
$|R'(S_1)| + |R'(S_t)| \leq 2 n$, and we can assume that $k \geq 1$.
Hence $R'$ must be a string described at point 1. It follows 
that 
\[
|R'| = 2n (t+1) +|R'(S_j)|
\]

Since $|R'| = 2n (t+1)+|R'(S_j)|$, then $R'(S_j)$
has length at least $k$.

We have described an OR-cross-composition of \LRS to itself. 
By Theorem \ref{teo:CrossComp}, it follows
that \LRS does not admit a polynomial kernel for parameters $|\Sigma|$, unless NP $\subseteq$ coNP/poly.
\end{proof}

We can complement the FPT algorithm of Section \ref{sec:FPT}, with a hardness of kernelization
for the same parameter.

\begin{corollary}
\label{cor:nokernel}
\LRS does not admit a polynomial kernel for parameters $r$, unless NP $\subseteq$ coNP/poly.
\end{corollary}
\begin{proof}
The result follows from Theorem \ref{th:nokernel}
and from the fact that $r \leq |\Sigma|$.
\end{proof}

\section{APX-hardness for Bounded Number of Occurrences}
\label{sec:APX2occ}

In this section, we show that \LRS
is hard even when the number of occurrences
of a symbol in the input string is bounded by two. 
We denote this restriction 
of the problem by \TwoLRS.
Notice that if the number of occurrences
of a symbol is bounded by one, then
the problem is trivial, 
as a solution of \LRS can have only runs of length one.

We prove
the result by giving a reduction from the
{\sc Maximum Independent Set problem on Cubic Graphs}
(\MISC), which is known to be APX-hard \cite{DBLP:journals/tcs/AlimontiK00}.
We recall the definition of \MISC:

\Pb{{\sc Maximum Independent Set problem on Cubic Graphs}
(\MISC)}{A cubic graph $G=(V,E)$.}{an independent set in $G$ of size at least $q$}{Output}

Given a cubic graph $G=(V,E)$, with $V=\{ v_1, \dots v_n \}$
and  $|E|=m$, we construct a corresponding instance $S$ of \TwoLRS (see Fig. \ref{fig:lreduction} for an example of our construction).
First, we define the alphabet 
$\Sigma$:

\[
\Sigma = \{  w_i: 1 \leq i \leq n \}
\cup  
\{  x_{i,j}^i, x_{i,j}^j, e_{i,j}^1, e_{i,j}^2: \{ v_i,v_j \} \in E, i < j \} \cup
\{  \sharp_{i,z}: 1 \leq i \leq m+n, 1 \leq z \leq 3 \}
\]

This alphabet is of size $n+4m+3(m+n) = 4n+7m$.

Now, we define a set of substrings of the 
instance $S$ of \TwoLRS that we are constructing.

\begin{itemize}

\item For each $v_i \in V$, $1 \leq i \leq n$, such that
$v_i$ is adjacent to $v_j$, $v_h$, $v_z$,
$1 \leq  j < h < z \leq n$,
we define a substring $S(v_i)$:
\[
S(v_i) = w_i x_{i,j}^i x_{i,h}^i 
x_{i,z}^i  w_i
\]
Notice that in the definition of $S(v_i)$ given above, we have assumed without loss of generality that $1 \leq i <  j < h <z \leq n$.
If, for example, $1 \leq j < i <h <z \leq n$, 
the symbol associated with $\{v_i,v_j\}$
is then $x_{j,i}^i$ and $S(v_i)$ is defined as 
follows:
\[
S(v_i) = w_i x_{j,i}^i x_{i,h}^i 
x_{i,z}^i  w_i
\]

\item For each edge $\{v_i,v_j\} \in E$,
with $1 \leq i<j \leq n$, we
define a substring $S(e_{ij})$:
\[
S(e_{ij}) = e_{i,j}^1 x_{i,j}^i e_{i,j}^2
e_{i,j}^1 x_{i,j}^j e_{i,j}^2
\]
\item We define separation substrings
$S_{Sep,i}$, with $1 \leq i \leq m+n$:
\[
S_{Sep,i} = \sharp_{i,1} \sharp_{i,2} \sharp_{i,3}
\]
\end{itemize}

Now, given the lexical ordering\footnote{$\{ v_i,v_j\} < \{ v_h, v_z\}$ (assuming $i<j$ and $h<z$) if and only if 
$i<h$ or $i=h$ and $j <z$.} 
of the edges of $G$,
the input string $S$ is defined as follows
(we assume that $\{v_1,v_z\}$ is the first edge
and $\{v_p,v_t\}$ is the last edge in the
lexicographic ordering of $E$):
\[
S = S(v_1) S_{Sep,1} S(v_2) S_{Sep,2}
\dots S(v_n) S_{Sep,n} S(e_{1,z}) S_{Sep,n+1} \dots
S(e_{p,t}) S_{Sep,n+m} 
\]

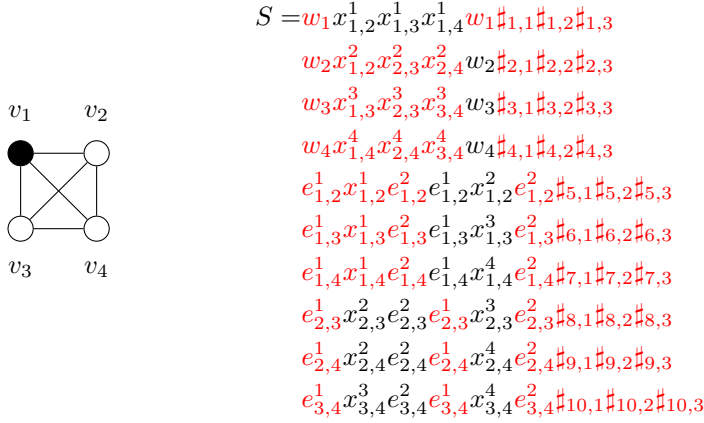
\begin{figure}
\centering
\begin{subfigure}{.2\textwidth}
\begin{tikzpicture}
\tikzstyle{every node}=[shape = circle,minimum size=0.5pt, draw]

\node[label=$v_1$,fill=black] (1) at (0,0) {}; 
\node[right of=1, label=$v_2$] (2) {}; 
\node[below of=1, label=below:$v_3$] (3) {}; 
\node[right of=3, label=below:$v_4$] (4) {}; 

\draw (1) -- (2);
\draw (1) -- (3);
\draw (1) -- (4);
\draw (2) -- (3);
\draw (2) -- (4);
\draw (3) -- (4);
\end{tikzpicture}
\end{subfigure}
\begin{subfigure}{.7\textwidth}
\begin{align*}
S = 
& {\color{red}w_1} x_{1,2}^1 x_{1,3}^1 x_{1,4}^1 {\color{red}w_1}
{\color{red}\sharp_{1,1}\sharp_{1,2}\sharp_{1,3}}\\
& {\color{red}w_2 x_{1,2}^2 x_{2,3}^2 x_{2,4}^2} w_2
{\color{red}\sharp_{2,1}\sharp_{2,2}\sharp_{2,3}}\\
& {\color{red}w_3 x_{1,3}^3 x_{2,3}^3 x_{3,4}^3} w_3
{\color{red}\sharp_{3,1}\sharp_{3,2}\sharp_{3,3}}\\
& {\color{red}w_4 x_{1,4}^4 x_{2,4}^4 x_{3,4}^4} w_4
{\color{red}\sharp_{4,1}\sharp_{4,2}\sharp_{4,3}}\\
& {\color{red}e_{1,2}^1 x_{1,2}^1 e_{1,2}^2} e_{1,2}^1 x_{1,2}^2 {\color{red}e_{1,2}^2}
{\color{red}\sharp_{5,1}\sharp_{5,2}\sharp_{5,3}}\\
& {\color{red}e_{1,3}^1 x_{1,3}^1 e_{1,3}^2} e_{1,3}^1 x_{1,3}^3 {\color{red}e_{1,3}^2}
{\color{red}\sharp_{6,1}\sharp_{6,2}\sharp_{6,3}}\\
& {\color{red}e_{1,4}^1 x_{1,4}^1 e_{1,4}^2} e_{1,4}^1 x_{1,4}^4 {\color{red}e_{1,4}^2}
{\color{red}\sharp_{7,1}\sharp_{7,2}\sharp_{7,3}}\\
& {\color{red}e_{2,3}^1} x_{2,3}^2 e_{2,3}^2 {\color{red}e_{2,3}^1} x_{2,3}^3 {\color{red}e_{2,3}^2}
{\color{red}\sharp_{8,1}\sharp_{8,2}\sharp_{8,3}}\\
& {\color{red}e_{2,4}^1} x_{2,4}^2 e_{2,4}^2 {\color{red}e_{2,4}^1} x_{2,4}^4 {\color{red}e_{2,4}^2}
{\color{red}\sharp_{9,1}\sharp_{9,2}\sharp_{9,3}}\\
& {\color{red}e_{3,4}^1} x_{3,4}^3 e_{3,4}^2 {\color{red}e_{3,4}^1} x_{3,4}^4 {\color{red}e_{3,4}^2}
{\color{red}\sharp_{10,1}\sharp_{10,2}\sharp_{10,3}}
\end{align*}
\end{subfigure}
\caption{A sample cubic graph (with 4 nodes and 6 edges) and the associated sequence. Black vertex of the graph corresponds to an independent set (of size 1 here), red symbols in the sequence correspond to the subsequence (of size $5\cdot 1 + 4 \cdot 3 + 3 \cdot 6 + 3 \cdot 10 = 65$).}\label{fig:lreduction}
\end{figure}

Now, we prove some properties on the string $S$.

\begin{lemma}
\label{lem:TwoOcc}
Let $G=(V,E)$ be an instance of \MISC
and let $S$ be the corresponding built instance
of \TwoLRS. Then $S$ contains at most two
occurrences for each symbol of $\Sigma$.
\end{lemma}
\begin{proof}
Notice that each symbol $w_i$, $1 \leq i \leq n$,
appears only in substring $S(v_i)$ of $S$.
Symbols $e_{i,j}^1$, $e_{i,j}^2$, with $\{ v_i,v_j\} \in E$
and $1 \leq i <j \leq n$, appear only in substring $S(e_{i,j})$ of $S$.
Each symbol $\sharp_{i,z}$, with $1 \leq i \leq m+n$
and $1 \leq z \leq 3$, appears only in
substring $S_{Sep,i}$ of $S$.
Finally, each symbol $x_{i,j}^i$,
with $\{ v_i,v_j\} \in E$, %
appears once in exactly two substrings
of $S$, namely
$S(v_i)$ and $S(e_{i,j})$.
\end{proof}

Now, we prove a property of solutions
of \TwoLRS relative to separation substrings.

\begin{lemma}
\label{lem:Sep}
Let $G=(V,E)$ be an instance of \MISC
and let $S$ be the corresponding instance
of \TwoLRS.
Given a run subsequence $R$ of $S$,
if $R$ does not contain some separation substring $S_{Sep,i}$,
with $1 \leq i \leq m+n$,
then there exists a run subsequence $R'$
of $S$ that contains $S_{Sep,i}$ and such
that $|R'| > |R|$. 
\end{lemma}
\begin{proof}
Notice that, since $R$ does not contain substring $S_{Sep,i}$, with $1 \leq i \leq m+n$,
it must contain
a run $r$ that connects two symbols that are on the left and on the right of $S_{Sep,i}$ in $S$,
otherwise $S_{Sep,i}$ can be added to 
$R$ increasing its length.
Since each symbol in $S$, hence also in $R$, has
at most two occurrences (see Lemma \ref{lem:TwoOcc}),
then $|r|=2$.
Then, starting from $R$, we can compute in 
polynomial time a run subsequence $R'$ by 
removing run $r$ and by adding substring 
$S_{Sep,i}$. Notice that, after the removal
of $r$, we can add $S_{Sep,i}$ since it 
contains three symbols each one having a single
occurrence in $S$.
Since $|S_{Sep,i}|=3$, it follows that $|S_{Sep,i}| > r$
and $|R'| > |R|$.
\end{proof}

Given a cubic graph $G=(V,E)$ and the corresponding instance
$S$ of \TwoLRS, a run subsequence $R$ of \TwoLRS on instance $S$ is called
\emph{canonical} if:

\begin{itemize}

\item for each $S_{Sep,i}$, $1 \leq i \leq m+n$,
$R$ contains $S_{Sep,i}$
(a substring denoted by $R_{Sep,i}$) 

\item for each $S(v_i)$, with $v_i \in V$, 
$R$ contains a substring $R(v_i)$ such that either
$R(v_i) = w_i w_i$  or it is a substring of length
$4$ ($w_i x_{i,j}^i x_{i,h}^i x_{i,z}^i$
or $x_{i,j}^i x_{i,h}^i x_{i,z}^i w_i$); moreover
if $\{ v_i,v_j \} \in E$, then 
at least one of $R(v_i)$ or $R(v_j)$
has length $4$

\item for each $S(e_{i,j})$, with $\{v_i,v_j\} \in E$, 
$R$ contains a substring $R(e_{i,j})$ 
such that $R(e_{i,j})$ is either of length $4$ ($e_{i,j}^1 x_{i,j}^i e_{i,j}^2 e_{i,j}^2$ or $e_{i,j}^1 e_{i,j}^1 x_{i,j}^j e_{i,j}^2$),
if one of $R(v_i)$, $R(v_j)$ has length $2$,
or of length $3$ ($e_{i,j}^1 e_{i,j}^1  e_{i,j}^2$ or
$e_{i,j}^1 e_{i,j}^2 e_{i,j}^2$).

\end{itemize}

\begin{lemma}
\label{lem:canonical}
Let $G=(V,E)$ be an instance of \MISC and let $S$ be the corresponding
instance of \TwoLRS. Given a run subsequence $R$ of $S$, %
we can compute in polynomial time  a canonical run subsequence  of $S$ of length at least $|R|$.
\end{lemma}
\begin{proof}
Consider a run subsequence $R$ of $S$. 
First, notice that
by Lemma \ref{lem:Sep} we assume that $R$ contains
each symbol $\sharp_{i,p}$, with $1 \leq i \leq n+m$
and $1 \leq p \leq 3$. 
We start by proving some bounds on the run subsequence
of $S(v_i)$ and $S(e_{i,j})$.

Consider a substring $R(v_i)$ of $S(v_i)$, $1 \leq i \leq n$. 
Each run subsequence of $S(v_i)$ can have length at most $4$, since $|S(v_i)| = 5$ and if run $w_i w_i$ belongs to $R(v_i)$,  then $R(v_i) = w_i w_i$. 
It follows that if $|R(v_i)| >2$, then it cannot contain the two occurrences of symbol $w_i$. %
Notice that the two possible run subsequences of length $4$ of $S(v_i)$ are $x_{i,j}^i x_{i,h}^i x_{i,z}^i  w_i$ and $ w_i x_{i,j}^i x_{i,h}^i x_{i,z}^i$.

Consider a run subsequence $R(e_{i,j})$ of $S(e_{ij}) = e_{i,j}^1 x_{i,j}^i e_{i,j}^2 e_{i,j}^1 x_{i,j}^j e_{i,j}^2$.
First, we prove that a run subsequence of $S(e_{ij})$ has length at most $4$ and in this case it must contain
at least one of $x_{i,j}^i $, $x_{i,j}^j$.
By its interleaved construction, at most one of runs $e_{i,j}^1 e_{i,j}^1$, $e_{i,j}^2 e_{i,j}^2$ can belong to $R(e_{i,j})$.
Moreover if $e_{i,j}^1 e_{i,j}^1$ 
($e_{i,j}^2 e_{i,j}^2$, respectively) belongs to 
$R(e_{i,j})$, then $|R(e_{i,j})|\leq 4$, since the longest
run in $S(e_{i,j})$ is then
$e_{i,j}^1 e_{i,j}^1 x_{i,j}^j e_{i,j}^2$
($e_{i,j}^1 x_{i,j}^i e_{i,j}^2 e_{i,j}^2$, 
respectively).
If none of runs $e_{i,j}^1 e_{i,j}^1$, $e_{i,j}^2 e_{i,j}^2$ belongs to $R(e_{i,j})$, then
$|R(e_{i,j})| \leq 4$, since 
$|S(e_{i,j})| = 6$; in this case both
$x_{i,j}^i$ and $x_{i,j}^j$ must be in 
$R(e_{i,j})$ to have $|R(e_{i,j})| = 4$.

Now, we compute a canonical run subsequence $R'$
of $S$ of length at least $|R|$.
Consider $R(v_i)$, $1 \leq i \leq n$,
and $R(e_{i,j})$, with $\{ v_i, v_j\} \in E$.

If $|R(v_i)|=4$, then define
$R'(v_i) =w_i x_{i,j}^i x_{i,h}^i x_{i,z}^i  $
(or equivalently define $R'(v_i) = x_{i,j}^i x_{i,h}^i x_{i,z}^i w_i$).

If $|R(v_i)| = 3$, then by construction of 
$S(v_i)$ at least two of $x_{i,j}^i$, 
$x_{i,h}^i$, $x_{i,z}^i$ belong to $R(v_i)$.
Then, at most one of $R(e_{i,j})$, $R(e_{i,h})$, $R(e_{i,z})$
can contain a symbol in $\{x_{i,j}^i, x_{i,h}^i, x_{i,z}^i\}$,
assume w.lo.g. that
$x_{i,j}^i$ belongs to $R(e_{i,j})$.
We define 
$
R'(v_i) =w_i x_{i,j}^i x_{i,h}^i x_{i,z}^i  
$
(or equivalently $R'(v_i) = x_{i,j}^i x_{i,h}^i x_{i,z}^i w_i$)
and $R'(e_{i,j}) = e_{i,j}^1 e_{i,j}^1  e_{i,j}^2$ (or equivalently 
$R'(e_{i,j}) = e_{i,j}^1 e_{i,j}^2 e_{i,j}^2$).
Since $|R(e_{i,j})| \leq 4$, 
we have that 
\[
|R'(e_{i,j})| \geq |R(e_{i,j})| -1
\]
and 
\[|R'(v_i)| = |R(v_i)| + 1
\]
It follows that the size of $R'$ is not decreased with
respect to the length of $R$.

If $|R(v_i)| = 2$, then define $R'(v_i) = w_i w_i$.

By construction of $R'$, either 
$|R'(v_i)|=4$ and 
\[
R'(v_i) = w_i x_{i,j}^i x_{i,h}^i x_{i,z}^i  
\text{  or  }
R'(v_i) = x_{i,j}^i x_{i,h}^i x_{i,z}^i w_i  
\]
or $|R'(v_i)|=2$ and 
\[
R'(v_i)= w_i w_i.
\]
Again the size of $R'$ is not decreased with
respect to the size of $R$.

In order to compute a canonical run subsequence,
we consider  an edge $\{ v_i,v_j \} \in E$
and the run subsequences $R'(v_i)$
and $R'(v_j)$ of $S(v_i)$, $S(v_j)$, 
respectively.
Consider the case that $R'(v_i)= w_i w_i$ and $R'(v_j) = w_j w_j$. Then
by construction $|R'(e_{i,j})| = 4$ and assume with
loss of generality that 
$R'(e_{i,j}) = e_{i,j}^1  e_{i,j}^1  x_{i,j}^j e_{i,j}^2$.
Now, we can modify $R'$ so that
\[
R'(v_i) = w_i x_{i,j}^i x_{i,h}^i x_{i,z}^i 
\]
by eventually removing $x_{i,h}^i$,  $x_{i,z}^i$
from $R'(e_{i,h})$ and $R'(e_{i,z})$. 
In this way, we decrease by at most
one the length of each of $R'(e_{i,h})$, $R'(e_{i,z})$
and we increase of two the length of $R'(v_i)$.
It follows that the length of $R'$ is not 
decreased by this modification.
By iterating this modification, we obtain
that for each edge $\{ v_i,v_j \} \in E$
at most one of $R'(v_i)$, $R'(v_j)$
has length two.

The run subsequence $R'$ we have built is then a canonical run subsequence of $S$ such that $|R'|\geq |R|$.
\end{proof}

Now, we are ready to prove the main results of the 
reduction.

\begin{lemma}
\label{lem:hard1}
Let $G=(V,E)$ be an instance of \MISC
and let $S$ be the corresponding instance
of \TwoLRS.
Given an independent set $I$ of size at least $q$ 
in $G$,
we can compute in polynomial time 
a run subsequence of $S$ of length
at least $5q+4(n-q)+3m+3(n+m)$.
\end{lemma}
\begin{proof}
We construct a subsequence run $R$ of $S$ 
as follows:
\begin{itemize}

\item For each $v_i \in I$, define
for the substring $S(v_i)$ the run subsequence $R(v_i) = w_i w_i$; %

\item For each $v_i \in V \setminus I$, define
for the substring $S(v_i)$ the run subsequence 
$R(v_i) = w_i x_{i,j}^i x_{i,h}^i x_{i,z}^i$;

\item For each $\{v_i,v_j\} \in E$,
if $v_i \in I$ (or $v_j \in I$, respectively)
define for the substring $S(e_{i,j})$ the run subsequence 
$R(e_{i,j})=e_{i,j}^1 x_{i,j}^i e_{i,j}^2 e_{i,j}^2$ 
($R(e_{i,j})=e_{i,j}^1 e_{i,j}^1 x_{i,j}^j e_{i,j}^2$, respectively);
if both $v_i, v_j \in V \setminus I$,
define for the substring $S(e_{i,j})$ the run subsequence 
$R(e_{i,j}) = e_{i,j}^1 e_{i,j}^1 e_{i,j}^2$ %

\end{itemize}

Moreover, $R$ contains each separation substring 
of $S$, denoted by $R_{Sep,i}$, $1 \leq i \leq n+m$.

First, we prove that $R$ is a run subsequence,
that is $R$ contains a single run for
each symbol in $\Sigma$. This property
holds by construction for each symbol
in $\Sigma$ having only occurrences
in $R(v_i)$, with $v_i \in V$,
$R(e_{i,j})$, with $\{v_i, v_j\} \in E$,
and $R_{Sep,i}$, $1 \leq i \leq n+m$.
What is left to prove is that
$x^i_{i,j}$ appears in at most
one of $R(v_i)$, with $v_i \in V$,
$R(e_{i,j})$, with $\{v_i, v_j\} \in E$.
Indeed, by construction, $R(e_{i,j})$ contains
$x^i_{i,j}$ only if $R(v_i) =w_i w_i $.
It follows that $R$ is a run subsequence of $S$.

Consider the length of $R$.
For each $v_i \in V \setminus I$, 
$R$ contains a run subsequence of $S(v_i)$ 
of length $4$.    
For each $v_i \in I$, $R$ contains a
run subsequence of length $2$.
For each $\{v_i,v_j\} \in E$,  
with $v_i, v_j \in V \setminus I$, $R$ contains
a run subsequence of $S(e_{i,j})$ of length $3$;
for each $\{v_i,v_j\} \in E$,  
with $v_i \in I$ or $ v_j \in I$,
$R$ contains a run subsequence of $S(e_{i,j})$
of length $4$.
Finally, each separation substring 
$R_{Sep,i}$, $1 \leq i \leq n+m$,
in $R$ has length $3$.
Hence the total length of $R$ is at least
$5q + 4(n-q)+3m+3(n+m)$ (by accounting,
for each $R(v_i)$ of length $2$,
the increasing of the length of
the three  run subsequences $R(e_{i,j})$, $R(e_{i,h})$,
$R(e_{i,z})$ from $3$ to $4$
 to $R(v_i)$). 
\end{proof}

\begin{lemma}
\label{lem:hard2}
Let $G=(V,E)$ be an instance of \MISC
and let $S$ be the corresponding instance
of \TwoLRS.
Given a run subsequence of $S$ of length
at least $5q+4(n-q)+3m+3(n+m)$, we compute
in polynomial time an independent of $G$ of size 
at least $q$.
\end{lemma}
\begin{proof}
Consider a run subsequence $R$ of $S$ of length
at least $5q+4(n-q)+3m+3(n+m)$. 
By Lemma \ref{lem:canonical},
we assume that $R$ is a canonical run subsequence of $S$.
It follows that we can define an independent $V'$ of
size at least $q$ in $G$ as follows:
\[
V' = \{ v_i: |R(v_i)| = 2\}
\]
By the definition of canonical run subsequence,
it follows that $V'$ is an independent set,
since if $|R(v_i)| = |R(v_j)|=2$, 
with $1 \leq i,j \leq n$,
then $\{v_i,v_j\} \notin E$.
Furthermore, by the definition of canonical 
run subsequence,
since $|R| \geq 5q+4(n-q)+3m+3(n+m)$, 
there are at least 
$q$ run subsequences $R(v_i)$, with $1 \leq i \leq n$, 
of length two such
that $|R(e_{i,j})|=|R(e_{i,h})|= |R(e_{i,z})|=4$,
with $\{v_i,v_j\}, \{v_i,v_h\}, \{v_i,v_z\} \in E$.
It follows that $|V'|\geq q$.
\end{proof}

Now, we can prove the main result
of this section.

\begin{theorem}
\label{thm:apx}
\TwoLRS is APX-hard. 
\end{theorem}
\begin{proof}
We have shown a reduction from \MISC to \TwoLRS. 
By Lemma \ref{lem:TwoOcc}, the instance of \TwoLRS we have built consists of a string with at most two occurrences for each symbol. 
We will now show that  this reduction is an L-reduction from \MISC to \TwoLRS (see Definition \ref{def:L-reduction}).

Consider an instance $I$ of \MISC
and a corresponding instance $I'$ of \TwoLRS.
Then, given any optimal solution
$opt(I')$ of \TwoLRS on instance $I'$,
by Lemma \ref{lem:hard2} it holds that

\[opt(I') \leq 5 \cdot opt(I) + 4( n-opt(I))+3m+3(n+m)
=
opt(I) + 7n + 6m
\]

In a cubic graph $G=(V,E)$, $|E| = \frac{3}{2}|V|$, hence $m = \frac{3}{2}n$. Furthermore, we can assume that an independent set has size at least $\frac{n}{4}$. 
Indeed, such an independent set $V'$ can be greedily computed as follows: pick a vertex $v$ in the graph, add it to $V'$ and delete $N[v]$ from $G$. 
At each step, we add one vertex in $V'$ and we delete at most $4$ vertices from $G$.

Since $m=\frac{3}{2}n$ and $n \leq 4 \cdot opt(I)$, we thus have that 
\[
opt(I') \leq opt(I) + 7n + 6m = opt(I) +
16\cdot n \leq opt(I) + 64 \cdot opt(I)
\]

and then $\alpha = 65$ in Definition \ref{def:L-reduction}.

Conversely, consider a solution $S'$ of length $5q + 4(n-q) + 3m + 3(n+m)$ of \TwoLRS on instance $I'$.
First, notice that by Lemma~\ref{lem:canonical}, we can
assume that $S'$ and $opt(I')$ are both canonical.
By Lemma \ref{lem:hard1}, if $opt(I) = p$, then
$opt(I') \geq 5 p + 4(n-p) + 3m + 3(n+m)$.
By Lemma \ref{lem:hard2}, 
starting from $S'$, we can compute in polynomial time a
solution $V'$ of \MISC on instance $I$, 
with $|V'| \geq q$. 
It follows that

\[
|opt(I) - |V'||  \leq
|opt(I) - q|  = |p - q| =  
\]
\[
|5p + 4(n-p) 
+3m + 3(n+m) - (5q + 4(n-q) + 3m+3(n+m))|
\leq 
\]
\[
|opt(I') - (5q + 4(n-q) + 3m +  3(n+m)) |
\]
Then $\beta = 1$ in Definition \ref{def:L-reduction}.
Thus we indeed have designed an L-reduction, therefore, the APX-hardness of \TwoLRS follows from the APX-hardness of \MISC~\cite{DBLP:journals/tcs/AlimontiK00} and from Theorem~\ref{thm:noptas}.

\end{proof}

\section{Conclusion}

In this paper, we deepen the understanding of the complexity of the recently introduced problem \LRS.
We show that the problem remains hard (even from the approximation point of view) also in the very restricted setting where each symbol occurs at most twice.
We also complete the parameterized complexity landscape.
From the more practical point of view, it is however unclear how our FPT algorithm could compete with implementations done in~\cite{DBLP:conf/wabi/SchrinnerGWSSK20}.

An interesting future direction is to further
investigate the approximation complexity of the
\LRS problem beyond APX-hardness. 
Note that a trivial $\min(|\Sigma|,occ)$-approximation algorithm 
($occ$ is the maximum number of occurrences
of a symbol in the input $S$)
can be designed by taking the solution
having maximum length between:
(1) a solution having one occurrence for each
symbol in $\Sigma$ and (2) a solution consisting of
the $a$-run of maximum
length, among each $a \in \Sigma$.
This leads to a $\sqrt{|S|}$-approximation algorithm.
Indeed, if the $a$-run of maximum length is greater than 
$\sqrt{|S|}$, then solution (2) has length at least $\sqrt{|S|}$, thus leading to the desired 
approximation factor.
If this is not the case, then each symbol in $\Sigma$ has
less then $\sqrt{|S|}$ occurrences,
thus a solution of \LRS on instance $S$ has at length
smaller than $|\Sigma| \sqrt{|S|}$.
It follows that (1) is a solution with the desired 
approximation factor.
We let for future work closing the gap between 
the APX-hardness and the $\sqrt{|S|}$-approximation factor of \LRS.

\bibliography{references}
\end{document}